%% file: main1.tex
\documentclass[11pt]{article}
\usepackage{fullpage}
\usepackage{amsmath,amsfonts,amsthm,amssymb}
\usepackage{url}
\usepackage{color}
\usepackage[usenames,dvipsnames,svgnames,table]{xcolor}
\usepackage[colorlinks=true, linkcolor=red, urlcolor=blue, citecolor=gray]{hyperref}
\usepackage[capitalise]{cleveref}
\usepackage{graphicx}
\usepackage{caption} 
\usepackage{hhline}
\usepackage{algorithm}
\usepackage{algorithmic}
\usepackage{natbib}
\usepackage{float}
\usepackage{subfigure}
\usepackage{booktabs} 

\newcommand{\Var}{\operatorname{Var}}
\newcommand{\eps}{\ensuremath{\epsilon}}

\DeclareMathOperator*{\E}{\mathbb{E}}

\newcommand{\R}{\mathbb{R}}

\newcommand{\norm}[1]{\|#1\|}

\usepackage{color}

\makeatletter

\newcommand{\specialcell}[2][c]{%
  \begin{tabular}[#1]{@{}c@{}}#2\end{tabular}}

\newtheorem*{rep@theorem}{\rep@title}
\newcommand{\newreptheorem}[2]{%
\newenvironment{rep#1}[1]{%
 \def\rep@title{#2 \ref{##1}}%
 \begin{rep@theorem}}%
 {\end{rep@theorem}}}
\makeatother
\newtheorem{theorem}{Theorem}
\newreptheorem{theorem}{Theorem}

\newtheorem{corollary}[theorem]{Corollary}
\newtheorem{lemma}[theorem]{Lemma}
\newtheorem*{lemma*}{Lemma}
\newreptheorem{lemma}{Lemma}

\newtheorem{claim}[theorem]{Claim}
\newtheorem{definition}[theorem]{Definition}

\input{bib_macros}

\date{}
	\author{
		Arturs Backurs\footnote{Toyota Technological Institute at Chicago. \texttt{backurs@ttic.edu}}
		\and Piotr Indyk\footnote{Massachusetts Institute of Technology. \texttt{indyk@mit.edu}}
		\and Cameron Musco\footnote{University of Massachusetts Amherst. \texttt{cmusco@cs.umass.edu}}
		\and Tal Wagner\footnote{Microsoft Research. \texttt{talw@mit.edu}}
	}
\title{Faster Kernel Matrix Algebra via Density Estimation}
\begin{document}

\maketitle

\begin{abstract}
We study fast algorithms for computing fundamental properties of a positive semidefinite kernel matrix $K \in \R^{n \times n}$ corresponding to $n$ points $x_1,\ldots,x_n \in \R^d$. In particular, we consider estimating the sum of kernel matrix entries, along with its top eigenvalue and eigenvector. 

We show that the sum of matrix entries can be estimated to $1+\epsilon$ relative error in time {\em sublinear} in $n$ and linear in $d$ for many popular kernels, including the Gaussian, exponential, and rational quadratic kernels. For these kernels, we also show that the top eigenvalue (and an approximate eigenvector) can be approximated to $1+\epsilon$ relative error in time \emph{subquadratic} in $n$ and linear in $d$. 

Our algorithms represent significant advances in the best known runtimes for these problems.
They leverage the positive definiteness of the kernel matrix, along with a recent line of work on efficient kernel density estimation.
\end{abstract}

\section{Introduction}

Kernels are a ubiquitous notion in statistics, machine learning, and other fields. A kernel is a function $k: \R^d\times \R^d \to \R$ that measures the similarity\footnote{This should be contrasted with {\em distance} functions that measure the {\em dissimilarity} between two vectors.}  between two $d$-dimensional vectors. Many statistical and machine learning methods, such as support vector machines, kernel ridge  regression and  kernel density estimation, rely on appropriate choices of kernels. A prominent example of a kernel function is the Radial Basis Function, a.k.a.\ Gaussian kernel, defined as
\begin{align*}
 k(x,y) = \exp(-\|x-y\|^2).
\end{align*}
Other popular choices include the Laplace kernel, exponential kernel, etc. See \citep{shawe2004kernel,hofmann2008kernel} for an overview.

Kernel methods typically operate using a kernel matrix. Given $n$ vectors $x_1, \ldots, x_n \in \R^d$, the kernel matrix $K \in \R^{n \times n}$ is defined as $K_{i,j}=k(x_i,x_j)$. For most popular kernels, e.g., the Gaussian kernel, $k$ is a positive definite function, and so $K$ is positive semidefinite (PSD). Furthermore, it is often the case that $K$'s entries are in the range [0,1], with $1$'s on the diagonal -- we assume this throughout. 

Although popular, the main drawback of kernel methods is their efficiency. Most kernel-based algorithms have running times that are at least quadratic in $n$; in fact, many start by explicitly materializing the kernel matrix $K$ in preprocessing.  This quadratic runtime is likely necessary as long as exact (or high-precision) answers are desired. Consider perhaps the simplest kernel problem, where the goal is to compute the sum of matrix entries, i.e.,  $s(K)=\sum_{i,j} K_{i,j}$. It was shown in \citet{backurs2017fine} that, for the Gaussian kernel, computing $s(K)$ up to $1 +\epsilon$ relative error requires $n^{2-o(1)}$ time under the Strong Exponential Time Hypothesis (SETH), as long as  the dimension $d$ is at least polylogarithmic in $n$, and $\epsilon=\exp(-\omega(\log^2 n))$. The same limitations were shown to apply  to kernel support vector machines, kernel ridge regression, and other kernel problems. 

Fortunately, the aforementioned lower bound does not preclude faster algorithms for larger values of $\epsilon$ (say, $\epsilon=\Theta(1)$). Over the last decade many such algorithms have been proposed. In our context, the most relevant ones are those solving kernel density evaluation~\citep{charikar2017hashing,backurs2018efficient,backurs2019space,siminelakis2019rehashing,charikar2020kernel}. Here,  we are given two sets of vectors $X=\{x_1, \ldots, x_m\}$, and $Y = \{y_1, \ldots, y_m$\}, and the goal is to compute the values of $k(y_i)= \frac{1}{m} \sum_j k(x_j,y_i)$ for $i=1, \ldots, m$. For the Gaussian kernel, the best known algorithm, due to \citet{charikar2020kernel}, estimates these values to $1 +\epsilon$ relative error in time $O(dm /(\mu^{0.173+o(1)}\epsilon^2))$, where $\mu$ is a lower bound on $k(x_i)$. 

Applying this algorithm directly to approximating the Gaussian kernel  sum yields a runtime of roughly $O(dn^{1.173+o(1)}/\epsilon^{2})$, since we can set $\mu = 1/n$ 
and still achieve $(1+\epsilon)$ approximation as $k(x_i,x_i)=1$ for all $x_i$. It is a natural question whether this bound can be improved and if progress on fast kernel density estimation can be extended to other fundamental kernel matrix problems.

\paragraph{Our results} 
In this paper we give much faster algorithms for approximating two fundamental quantities: the kernel matrix sum and the kernel matrix top eigenvector/eigenvalue. 
Consider a kernel matrix $K$ induced by $n$ points $x_1,\ldots,x_n \in \R^d$ and a kernel function with values in $[0,1]$ such that the matrix $K$ is PSD and has $1$'s on the diagonal. Furthermore, suppose that the kernel $k$ is supported by a kernel density evaluation algorithm with running time of the form $O(dm/(\mu^p \epsilon^2))$ for $m$ points, relative error $(1+\epsilon)$, and density lower bound $\mu$. Then we give: 
\begin{enumerate}
\item An algorithm for $(1+\epsilon)$-approximating $s(K)$ in time:
\begin{align*}
O \left (dn^{2 + 5p \over 4+2p}/\epsilon^{8+6p \over 2+p} \cdot \log^2 n \right).
\end{align*}
For many popular kernels the above runtime is {\em sublinear} in $n$ -- see Table \ref{tab:main}. Our algorithm is based on subsampling $O(\sqrt{n})$ points from $x_1, \ldots, x_n$ and then applying fast kernel density evaluation to these points.
We complement the algorithm with a (very simple) lower bound showing that sampling $\Omega(\sqrt{n})$ points 
 is necessary to estimate $s(K)$ up to a constant factor. This shows that our 
  {\em sampling} complexity is optimal. 

\item An algorithm that returns an approximate top eigenvector $z \in \R^{n}$ with $\norm{z}_2 = 1$ and $z^T K z \ge (1-\epsilon) \cdot \lambda_1(K)$, where $\lambda_1(K)$ is $K$'s top eigenvalue, running  in time:
\begin{align*}
O \left(\frac{dn^{1+p} \log(n/\epsilon)^{2+p}}{\epsilon^{7+4p}} \right ).
\end{align*}
For many popular kernels, this runtime is  subquadratic in $n$ -- see Table \ref{tab:main}. This is the first  subquadratic time algorithm for top eigenvalue approximation and a major improvement over forming the full kernel matrix. By a simple argument, $\Omega(dn)$ time is necessary even for constant factor approximation, and thus our algorithm is within an $\tilde O(n^p)$ factor of optimal. Our algorithm is also simple and practical, significantly outperforming baseline methods empirically -- see Section \ref{sec:exp}.
%
%
\end{enumerate}
%
\begin{table}[!h]
\footnotesize
\begin{center}
\begin{tabular}{|c|c|c|c|c|c|}
\hline
Kernel & $k(x,y)$ &KDE algorithm & KDE runtime & Kernel Sum & Top Eigenvector \\
\hline
 Gaussian  & $e^{-\|x-y\|_2^2}$ & \citet{charikar2020kernel}  & $dm /(\mu^{0.173+o(1)})$ &  $\tilde{O}(d n^{0.66})$ & $\tilde O(d n^{1.173+o(1)})$ \\
 Gaussian  &  & \specialcell{\citet{greengard1991fast} \\(see appendix)}  & $m \log(m)^{O(d)}$ & $n^{0.5} \log(n)^{O(d)}$ & $n \log(n)^{O(d)}$ \\
 Exponential & $e^{-\|x-y\|_2}$ & \citet{charikar2020kernel}  & $dm /(\mu^{0.1+o(1)})$ & $\tilde{O}(d n^{0.6})$ & $\tilde O(d n^{1.1+o(1)})$ \\
  Rational quadratic & $\frac{1}{ (1+\|x-y\|_2^2)^{\beta}}$ & \citet{backurs2018efficient} & $d$ & $\tilde{O}(d n^{0.5})$ & $\tilde O(dn)$ \\
  \specialcell{All of the above\\ (lower bound)}  & - & - & - & $\Omega(d n^{0.5})$ & $\Omega(dn)$ \\
\hline
\end{tabular}
\end{center}
\normalsize
\caption{Instantiations of our main results, giving sublinear time kernel sum approximation and subquadratic time top eigenvector approximation. All running times are up to  polylogarithmic factors  assuming constant accuracy parameter $\epsilon$ and kernel parameter $\beta$. 
The KDE runtime depends on $m$, the number of query points and $\mu$, a lower bound on the density for each query point.}\label{tab:main}
\end{table}

\paragraph{Application} An immediate application of our kernel sum algorithm is a faster algorithm for estimating the {\em kernel alignment}~\citep{cristianini2002kernel}, a popular measure of similarity between kernel matrices. Given $K$ and $K'$, the {\em alignment} between $K$ and $K'$ is defined as
\begin{align*}
    {\hat{A}(K, K')= \frac{\langle K,K'\rangle}{\sqrt{\langle K,K\rangle \langle K',K'\rangle}},}
    \end{align*}
where $\langle K,K'\rangle = \sum_{i,j} K_{i,j} K'_{i,j}$ is the inner product between the matrices $K$ and $K'$ interpreted as vectors. Our algorithm yields an efficient algorithm for estimating $\hat{A}(K, K')$ as long as the product kernels $K \circ K$, $K' \circ K'$ and $K \circ K'$ are supported by fast kernel density evaluation algorithms as described earlier. This is the case for e.g., the Gaussian or Laplace kernels.

\paragraph{Related work}
The problem of evaluating kernel densities, especially for the Gaussian kernel, has been studied extensively. In addition to the recent randomized algorithms discussed in the introduction, there has been a considerable amount of work on algorithms in low dimensional spaces, including~\cite{greengard1991fast,yang2003improved,lee2006dual,lee2009fast,march2015askit}. We present a streamlined version of the Fast Gauss Transform algorithm of~\citet{greengard1991fast} in the appendix. In addition, there has been a considerable effort designing {\em core-sets} for this problem~\citep{phillips2018improved,phillips2018near}. 

The sum of kernel values can be viewed as a similarity analog of the sum of pairwise {\em distances} in metric spaces. The latter quantity can be approximated in time linear in the number $n$ of points in the metric space~\citep{indyk1999sublinear,chechik2015average,cohen2018clustering}. Note that it is not possible to achieve an $o(n)$-time algorithm for this problem, as a single point can greatly affect the overall value. To the best of our knowledge, our algorithms for the kernel matrix sum are the first that achieve sublinear in $n$ running time and give $O(1)$-approximation for a nontrivial kernel problem. 

Computing the top eigenvectors of a kernel matrix is a central problem -- it is the primary operation behind kernel principal component analysis \citep{scholkopf1997kernel}. Projection onto these eigenvectors also yields an optimal low-rank approximation to the kernel matrix, which can be used  Low-rank approximation is widely used to approximate kernel matrices, to speed up kernel learning methods \citep{williams2001using,fine2001efficient}. Significant work has focused on fast low-rank approximation algorithms for kernel matrices or related distance matrices \citep{musco2017recursive,musco2017sublinear,bakshi2018sublinear,indyk2019sample,bakshi2020testing}. These algorithms have runtime scaling roughly linearly in $n$. However, they do not give any nontrivial approximation to the top eigenvalues or eigenvectors of the kernel matrix themselves, unless we assume that the matrix is near low-rank. To the best of our knowledge, prior to our work, no subquadratic time approximation algorithms for the top eigenvalue were known, even for $\Theta(1)$ approximation.

\paragraph{Our Techniques: Kernel Sum}
We start by noting that, via a Chernoff bound, a simple random sampling of kernel matrix entries provides the desired estimation in linear time.

\begin{claim}\label{clm:unisampling}
For a positive definite kernel $k: \R^d \times \R^d \rightarrow [0,1]$ with $k(x,x) = 1$ $\forall\, x$,   uniformly sample $t = O\left (\frac{n \log(1/\delta)}{\epsilon^2}\right )$ off-diagonal entries of $K$, $K_{i_1,j_1},...,K_{i_t,j_t}$ and let $\tilde s(K) = n + \frac{n(n-1)}{t} \cdot \sum_{\ell=1}^t K_{i_\ell,j_\ell}$. Then with probability $\ge 1-\delta$, $\tilde s(K) \in (1\pm \epsilon) \cdot s(K)$. 
\end{claim}

Our goal is to do better, giving $1\pm \epsilon$ approximation to $s(K)$ in sublinear time. We achieve this by (a) performing a more structured random sampling, i.e., sampling a principal submatrix as opposed to individual entries, and (b) providing an efficient algorithm for processing this submatrix.
Subsampling the matrix requires a more careful analysis of the variance of the estimator. To accomplish this, we use the fact that the kernel matrix is PSD, which implies that its  ``mass'' cannot be too concentrated. For the second task, we use fast kernel density estimation algorithms, combined with random row sampling to reduce the running time.

\paragraph{Our Techniques: Top Eigenvector}

Our algorithm for top eigenvector approximation is a variant on the classic power method, with fast approximate matrix vector multiplication implemented through kernel density evaluation. Kernel density evaluation on a set of $n$ points with corresponding kernel matrix $K \in \R^{n \times n}$, can be viewed as approximating the vector $K z \in \R^n$ where $z(i) = 1/n$ for all $i$. Building on this primitive, it is possible to implement approximate multiplication with general $z \in \R^n$ \citep{charikar2017hashing}. We can then hope to leverage work on the `noisy power method', which approximates the top eigenvectors of a matrix using just approximate matrix vector multiplications with that matrix \citep{hardt2014noisy}. However, existing analysis assumes {\em  random} noise on each matrix vector multiplication, which does not align with the top eigenvector. This cannot be guaranteed in our setting. Fortunately, we can leverage additional structure: if the kernel $k$ is non-negative, then by the Perron-Frobenius theorem, $K$'s top eigenvector is entrywise non-negative. This ensures that, if our noise in approximating $Kz$ at each step of the power method is entrywise non-negative, then this noise will have non-negative dot product with the top eigenvector. We are able to guarantee this property, and show convergence of the method to an approximate top eigenvector, even when the error might align significantly with the top eigenvector. 

\section{Preliminaries}\label{sec:prelim}


Throughout, we focus on {\em nice} kernels satisfying:
\begin{definition}[Nice Kernel Function]\label{def:nice} A kernel function $k: \R^d \times \R^d  \rightarrow [0,1]$ is \emph{nice} if it is positive definite and satisfies
$k(x,x)=1$ for all $x \in \R^d$.
\end{definition}
Many popular kernels, such as the Gaussian kernel, the exponential kernel and the rational quadratic kernel described in the introduction,  are indeed nice. We also assume that $k$ admits a fast KDE algorithm. Specifically:
\begin{definition}[Fast KDE]\label{def:fast}
A kernel function $k$ admits a $O(dm/(\mu^p \epsilon^2))$ time KDE algorithm, if given a set of $m$ points $x_1, \ldots, x_n \in \R^d$, we can process them in $O(dm/(\mu^p \epsilon^2))$ time for some $p \ge 0$ such that we can answer queries of the form ${1 \over n}\sum_i k(y,x_i)$ up to $(1+\epsilon)$ relative error in $O(d/(\mu^p \epsilon^2))$ time for any query point $y \in \R^d$ with probability $\geq 2/3$ assuming that ${1 \over n}\sum_i k(y,x_i) \geq \mu$.
\end{definition}
Note that for a kernel function satisfying Def. \ref{def:fast}, evaluating $k(y_j)$ for $m$ points $Y = \{y_1,\ldots, y_m\}$ against $m$ points $X = \{x_1,\ldots,x_m\}$ requires $O(dm/(\mu^p \epsilon^2))$ total time.

\section{Sublinear Time Algorithm for Kernel Sum}\label{sec:kernelsum}

Our proposed kernel sum approximation algorithm will sample a set $A$ of $s = \Theta(\sqrt{n})$ input points and look at the principal submatrix $K_A$ of $K$ corresponding to those points. We prove that the sum of off-diagonal entries in $K_A$ (appropriately  scaled) is a good estimator of the sum of off-diagonal entries of $K$. Since for a nice kernel, the sum of diagonal entries is always $n$, this is enough to give a good estimate of the full kernel sum $s(K)$. Furthermore, we show how to estimate the sum of off-diagonal entries of $K_A$ quickly via kernel density evaluation, in $ds^{2-\delta}/\eps^{O(1)} = dn^{1-\delta/2}/\eps^{O(1)}$ time for a constant $\delta>0$. 
Overall this yields:
\begin{theorem}
\label{t:main}
    Let $K \in \R^{n \times n}$ be a kernel matrix defined by a set of $n$ points and a nice kernel function $k$ (Def. \ref{def:nice}) that admits $O(dm/(\mu^p \epsilon^2))$ time approximate kernel density evaluation (Def. \ref{def:fast}). After sampling a total of $O(\sqrt{n}/\eps^2)$ points, we can in $O \left (dn^{2 + 5p \over 4+2p}\log^2(n)/\eps^{8+6p \over 2+p} \right )$ time approximate the sum of entries of $K$ within a factor of $1+\eps$ with high probability $1-1/n^{\Theta(1)}$.
\end{theorem}
For any square matrix $K$, let  $s_o(K)$ be  the sum of off diagonal entries.
Crucial to our analysis will be the lemma:
\begin{lemma}[PSD Mass is Spread Out]\label{lem:spread}
Let $K \in \R^{n \times n}$ be a PSD matrix with diagonal entries all equal to $1$.
Let $s_{o,i}(K)$ be the sum of off diagonal entries in the $i^{th}$ row of $K$. If $s_o(K) \ge \epsilon n$ for some $\eps\le 1$, then $\forall\, i$:
\begin{align*}
s_{o,i}(K) \le 2\sqrt{s_0(K)/\epsilon}.
\end{align*}
\end{lemma}
Lemma \ref{lem:spread} implies that if the off-diagonal elements contribute significantly  to $s(K)$ (i.e., $s_0(K) \ge \epsilon n$), then the off diagonal weight is spread relatively evenly across at least $\Omega(\sqrt{\epsilon \cdot s_o(K)}) = \Omega(\sqrt{n})$ rows/columns. This allows our strategy of sampling a principal submatrix with just $\Theta(\sqrt{n})$ rows to work.
\begin{proof}
Assume for the sake of contradiction that there is a row with $s_{o,i}(K) > 2 \sqrt{s_0(K)/\epsilon}$. Let $x$ be the vector that has value $-\sqrt{\frac{s_0(K)}{\epsilon}}$ at index $i$ and $1$ elsewhere. Then:
\begin{align*}
x^T K x \le  \frac{s_o(K)}{\epsilon} - \frac{4 \cdot s_o(K)}{\epsilon} + s_o(K)+ n \le -\frac{s_o(K)}{\epsilon},
\end{align*}
where the last inequality follows from the assumptions that $s_0(K) \ge \epsilon n$ and $\eps\leq 1$. The above contradicts $K$ being PSD, completing the lemma.
\end{proof}

\subsection{Our Estimator}


For a subset $A \subseteq [n]$, let $K_A$ be the corresponding kernel matrix (which is a principal submatrix of $K$). Suppose that $A$ is chosen by adding every element $i \in [n]$ to $A$ independently at random with probability $p$ (we will later set $p=1/\sqrt n$). Then $Z \triangleq n+s_o(K_A)/p^2$ is an unbiased estimator of $s(K)$. That is, $\E[Z]=s(K)$. We would like to show that the variance $\Var[Z]$ is small. In fact, in Lemma \ref{lem:var} below, we show that $\Var[Z] = O(s(K)^2)$. Thus,  taking $\Var[Z] / (\eps^2\E[Z]^2) = O(1/\eps^2)$ samples of $Z$ and returning the average yields a $1+\eps$ approximation of $s(K)$ with a constant probability.  To amplify the probability of success to $1-\delta$ for any $\delta>0$,  we  take the median of $O(\log(1/\delta))$ estimates and apply Chernoff bound in a standard way. Our variance bound follows:


\begin{lemma}\label{lem:var}
    $\Var[Z] =  O(s(K)^2).$
\end{lemma}
\begin{proof}
Let $Z_o\triangleq Z-n=s_o(K_A)/p^2$. 
\begin{align*}
	\Var[Z]=\Var[Z_o]&=\E[Z_o^2]-\E[Z_o]^2\\
	&\leq \E[Z_o^2]\\
	&={1 \over p^4} \E[s_o(K_A)^2]\\
	&={1 \over p^4}\left( p^2 \sum_{\substack{i,j \in [n]\\i \neq j}} K_{i,j}^2+2p^3 \sum_{\substack{i,j,j'\in [n]\\|\{i,j,j'\}|=3}}K_{i,j}K_{i,j'} + p^4 \sum_{\substack{i,j,i',j' \in [n]\\|\{i,j,i',j'\}|=4}}K_{i,j}K_{i',j'}  \right).
\end{align*}

We upper bound each term of the above expression separately. We start with the first term:
\begin{align*}
	p^{-2} \sum_{\substack{i,j \in [n]\\i \neq j}} K_{i,j}^2 = n \sum_{\substack{i,j \in [n]\\i \neq j}} K_{i,j}^2 \leq n \cdot s(K) \leq s(K)^2,
\end{align*}
where in the equality we set $p=1/\sqrt{n}$ and in the first inequality we use the fact that $0\leq K_{i,j}\leq 1$ for every $i,j$. We have $s(K) \ge n$ since all diagonal entries are $1$, giving the last inequality.

We upper bound the third term:
$$
	\sum_{\substack{i,j,i',j' \in [n]\\|\{i,j,i',j'\}|=4}}K_{i,j}K_{i',j'} \leq \sum_{i,j,i',j' \in [n]}K_{i,j}K_{i',j'}=s(K)^2.
$$

To upper bound the second term we consider two cases.

\paragraph{Case $s_o(K)\leq n$.} In this case we can use Lemma~\ref{lem:spread} with $\eps=s_o(K)/n$ to conclude that 

$$s_{o,i}(K) \le 2 \sqrt{s_o(K)/\eps}=2\sqrt{s_o(K)\cdot  n/s_o(K)} = 2\sqrt{n}$$

With this bound we can bound the second term by

$$p^{-1} \sum_{\substack{i,j,j'\in [n]\\|\{i,j,j'\}|=3}}K_{i,j}K_{i,j'} \leq p^{-1} \cdot  s_o(K)  \max_i s_{o,i}(K) \leq 2 n s_o(K)\leq 2s(K)^2,$$

where we substituted $p=1/\sqrt{n}$.

\paragraph{Case $s_o(K)>n$.} In this case we proceed as follows.
\begin{align*}
	p^{-1} \sum_{\substack{i,j,j'\in [n]\\|\{i,j,j'\}|=3}}K_{i,j}K_{i,j'} & \leq p^{-1} \sum_i \left(\sum_{j \, : \, j\neq i}K_{i,j}\right)^2\\
	&=p^{-1} \sum_i s_{o,i}(K)^2.
\end{align*}

Since $\sum_i s_{o,i}(K)$ is fixed (and is equal to $s_o(K)$), the last expression is maximized when some $s_{o,i}$ take as large values as possible and the rest are set to $0$. Since  $s_{o,i}(K)\leq t\triangleq 2\sqrt{s_o(K)}$ (by Lemma~\ref{lem:spread} with $\eps=1$), we have that in the worst case for $s_o(K)/t$ values of $i$ we have $s_{o,i}=t$ and $s_{o,i}=0$ for the rest of $i$. Therefore,
\begin{align*}
	p^{-1} \sum_i s_{o,i}(K)^2&\leq p^{-1} t^2 {s_o(K) \over t}\\
	&={2 \over p\sqrt{s_o(K)}}s_o(K)^2\\
	&\leq 2 s_o(K)^2\\
	&\leq 2 s(K)^2,
\end{align*}
where we use $s_o(K)> n$ and $p={1\over \sqrt{n}}$.
\end{proof}

\subsection{Approximating the Value Of the Estimator}

To turn the argument from the previous section into an algorithm, we need to approximate the value of $Z=n+s_o(K_A)/p^2$ for $p=1/\sqrt{n}$ efficiently. It is sufficient to efficiently approximate $Z = n +n\cdot s_o(K_A)$ when $s_o(K_A) = \Omega(\eps)$, as otherwise the induced loss in approximating $s(K)$ is  negligible since we always have $s(K) \ge n$.

Let $K'$ be a kernel matrix of size $m \times m$ for which $s_o(K')\geq \Omega(\eps)$. We show that for such a kernel matrix it is possible to approximate $s_o(K')$ in time $m^{2-\delta}$ for a constant $\delta>0$. This is enough to yield a sublinear time algorithm for estimating $Z$ since $K_A$ is $m \times m$ with $m\approx pn=\sqrt{n}$.

\paragraph{A simple algorithm.} We note that it is sufficient to approximate the contribution to $s_o(K')$ from the rows $i$ for which the sum of entries $s_{o,i}$ is $\Omega(\eps/m)$, as the contribution from the rest of the rows in negligible under the assumption that $s_o(K') = \Omega(\epsilon)$. So fix an $i$ and assume that $s_{o,i} \geq \Omega(\eps/m)$. To estimate $s_{o,i}$, we use a kernel density evaluation algorithm.
Our goal is to approximate
$$
    s_{o,i} = \sum_{j:j\neq i}K'_{i,j} = \sum_{j:j \neq i}k(x_i,x_j).
$$
The approach is to first process the points $x_1, \ldots, x_m$ using the  algorithm for the KDE. The KDE query algorithm then allows us to answer queries of the form ${1 \over m}\sum_j k(q,x_j)$ for an arbitrary query point $q$ in time $O(d/(\mu^p \eps^2))$, where $\mu$ is a lower bound on ${1 \over m}\sum_j k(q,x_j)$. So we  set $\mu=\Omega(\eps/m^2)$ and query the KDE data structure on all $q = x_1, \ldots, x_m$. 

The above does not quite work however -- to estimate the off-diagonal sum we need to answer queries of the form ${1 \over m}\sum_{j:j \neq i}k(x_i,x_j)$ instead of ${1 \over m}\sum_{j}k(x_i,x_j)$. This could be solved if the KDE data structure were dynamic, so that we could remove any point $x_i$ from it. Some of the data structures indeed have this property. To provide a general reduction, however, we avoid this requirement by building several ``static'' data structures and then answering a single query of the form ${1 \over m}\sum_{j:j \neq i}k(x_i,x_j)$ by querying $O(\log m)$ static data structures. Assume w.l.o.g.\ that $m$ is an integer power of $2$. Then we build a data structure for points $x_1, \ldots, x_{m/2}$ and another for $x_{m/2 + 1}, \ldots, x_m$. We also build $4$ data structures for sets $x_1, \ldots, x_{m/4}$ and $x_{m/4+1}, \ldots, x_{m/2}$, and $x_{m/2+1}, \ldots, x_{3m/4}$, and $x_{3m/4+1}, \ldots, x_m$. And so forth for $\log m$ levels. Suppose that we want to estimate ${1 \over m}\sum_{j:j \neq 1}k(x_1,x_j)$. For that we query the data structures on sets $x_{m/2+1}, \ldots, x_m$ and $x_{m/4+1}, \ldots, x_{m/2}$, and $x_{m/8+1}, \ldots, x_{m/4}$ and so forth for a total of $\log m$ data structures -- one from each level. Similarly we can answer queries for an arbitrary $i$. The threshold $\mu$ for all the data structures is the same as before: $\mu=\Omega(\eps/m^2)$.

Since we need to query $O(\log m)$ data structures for every $x_i$ and we also need to amplify the probability of success to, say, $1 - 1/m^2$. Thus,  the final runtime of the algorithm is
$$
    O(m \cdot d/(\mu^p\eps^2)\log^2 m) = O(dm^{1+2p}/\eps^{2+p}\log^2 m).
$$



\paragraph{A faster algorithm.} We note that in the previous case, if $s_{o,i}=\Theta(\eps/m)$ for every $i$, then we can approximate $s_o(K')$ efficiently by sampling a few $i$, evaluating the corresponding $s_{o,i}$ exactly (in $O(dm)$ time) and returning the empirical mean of the evaluated $s_{o,i}$. This works since the variance is small. There can be, however, values of $i$ for which $s_{o,i}$ is large. For these values of $i$, we can run a kernel density evaluation algorithm. More formally, we define a threshold $t>0$ and run a kernel density evaluation algorithm on every $i$ with $\mu=t/m^2$ (similarly as in the previous algorithm). This reports all $i$ for which $s_{o,i}\geq t/m^2$. This takes time
$$
    O(dm\log^2m \cdot 1/(\mu^p \eps^2))=O(dm^{1 + 2p}\log^2m/(\eps^2 t^p)).
$$
Let $I$ be the set of remaining $i$.
To estimate the contribution from $s_{o,i}$ with $i \in I$, we repeatedly sample $i \in I$ and evaluate $s_{o,i}$ exactly using the linear scan, and output the average of the evaluated $s_{o,i}$ scaled by $|I|$ as an estimate of the contribution of $s_{o,i}$ from $i \in I$.
Since we can ignore the contribution from $i \in I$ with $s_{o,i}\leq o(\eps/m)$, we can assume that $t/m\geq s_{o,i}\geq \Omega(\eps/m)$ for every $i \in I$ and then
$$
    \Var_i[s_{o,i}]/(\eps^2 (\E_i[s_{o,i}])^2) \leq O(t^2/\eps^4)
$$
samples are sufficient to get a $1+\eps$ approximation. This step takes time $O(t^2 dm/\eps^4)$. The final runtime is
\begin{align*}
    O(dm^{1 + 2p}\log^2m/(\eps^2 t^p)+dt^2 m/\eps^4) &=  O(dm^{2 + 5p \over 2+p}\log^2(m)/\eps^{4+4p \over 2+p})
\end{align*}
by setting $t=m^{2p \over 2+p}\eps^{2 \over 2+p}$. Since $m = \Theta(\sqrt{n})$ with high probability, we achieve $O(dm^{2 + 5p \over 4+2p}\log^2(m)/\eps^{4+4p \over 2+p})$ runtime for approximating the random variable $Z = n + n \cdot s_o(K_A)$. Since we evaluate the random variable $Z$ $O(1/\eps^2)$ times, the final runtime to approximate the sum of entries of the kernel matrix $K$ within a factor of $1+\eps$ is $O(dn^{2 + 5p \over 4+2p}\log^2(n)/\eps^{8+6p \over 2+p})$.

\subsection{Sample Complexity Lower Bound}

We next prove a lower bound, which shows that sampling just $O(\sqrt{n})$ data points, as is done in our algorithm, is optimal up to constant factors.
\begin{theorem}\label{thm:sample}
    Consider any nice kernel $k$ such that $k(x,y) \to 0$ as $\|x-y\| \to \infty$. In order to estimate $\sum_{i,j}^n k(x_i,x_j)$ within any constant factor, we need to sample at least $\Omega(\sqrt n)$ points from the input $x_1, \ldots, x_n$.
\end{theorem}
\begin{proof}
    Suppose that we want to approximate $s(K) = \sum_{i, j}^n k(x_i,x_j)$ within a factor of $C>0$.
    
    Consider two distributions of $x_1, \ldots, x_n$. For the first distribution $x_1, \ldots, x_n$ are sampled independently from a large enough domain such that $k(x_i,x_j)\approx 0$ for all $i \neq j$ with high probability. In this case $s(K) \approx n$. For the other distribution we sample again $x_1, \ldots, x_n$ independently at random as before and then 
    sample $s_1, \ldots, s_{\sqrt{2Cn}}$ from $1, \ldots, n$ without repetition. Then we set $x_{s_t}=x_{s_1}$ for $t = 2, \ldots, \sqrt{2Cn}$. For this distribution $s(K) \approx 2Cn$. To distinguish between these two distributions we need to sample at least $\Omega(\sqrt{n/C})=\Omega(\sqrt n)$ points from $x_1, \ldots, x_n$.
\end{proof}

\subsection{Application to Kernel Alignment}

Our algorithm can be immediately used to estimate  the value of the kernel alignment $\hat{A}(K, K')$ as defined in the introduction. The only requirement is that  the submatrices of product kernels $K \circ K$, $K' \circ K'$ and $K \circ K'$ are supported by fast kernel density evaluation algorithms. We formalize this as follows. Let $C$ be a set of nice kernels defined over pairs of points in $\R^d$. For any two kernel functions  $k, k' \in C$, the {\em product kernel} $k \circ k': \R^{2d} \times \R^{2d} \to [0,1]$ is such that for any $p,q,p',q' \in \R^d$ we have
\begin{align*}
(k \circ k')( (p,p'), (q,q'))=k(p,q) \cdot k'(p',q').
\end{align*}

\begin{definition}
Let $C=C^1, C^2, \ldots$ be a sequence of sets of nice kernels, such that $C^d$ is defined over pairs of points in $\R^d$. We say that $C$ is {\em closed under product} if for any two $k, k' \in C^d$, the product kernel $k \circ k'$ belongs to $C^{2d}$.
\end{definition}

It is immediate that the Gaussian kernel, interpreted as a sequence of kernels for different values of the dimension $d$, is closed under product. Thus we obtain the following:
\begin{corollary}
Given two Gaussian kernel matrices $K, K' \in \R^{n \times n}$, each defined by  a set of $n$ points in $\R^d$, and $\epsilon \in (0,1)$, $\hat{A}(K,K')$ can be estimated to $1 \pm \epsilon$ relative error in time $O(dn^{0.66}/\epsilon^{O(1)} \log^2 n)$ with high probability $1-1/n^{\Theta(1)}$. 
\end{corollary}

\section{Subquadratic Time Top Eigenvector}\label{sec:topeigvec}

We now present our top eigenvector approximation algorithm, which is a variant on the `noisy power method' with approximate matrix vector multiplication implemented through kernel density evaluation. 
Existing analysis of the  noisy power method assumes random noise on each matrix vector multiplication, which has little correlation with the top eigenvector  \citep{hardt2014noisy,balcan2016improved}. This prevents this top direction from being `washed out' by the noise. In our setting this cannot be guaranteed -- the noise distribution arising from implementing matrix multiplication with $K$ using kernel density evaluation is complex.

To avoid this issue, we use that since the kernel $k$ is nice, $K$ is entrywise non-negative, and by  the Perron-Frobenius theorem, so is its top eigenvector. Thus, if our noise in approximating $Kz$ is entrywise non-negative (i.e., if we overestimate each weighted kernel density), then the noise will have non-negative dot product with the top eigenvector, and will not wash it out, even if it is highly correlated with it. 

We formalize this analysis in Theorem  \ref{t:noisePower}, first giving the required approximate matrix vector multiplication primitive that we use in Definition \ref{def:mvm}. We give a full description of our noisy power method variant in Algorithm \ref{alg:power}. In Section \ref{sec:mvm} we discuss how to implement the matrix vector multiplication primitive efficiently using existing KDE algorithms.

\begin{definition}[Non-negative Approximate Matrix Vector Multiplication]\label{def:mvm}
An $\epsilon$-non-negative approximate MVM algorithm for a matrix $K \in \R^{n \times n}$ takes as input a non-negative vector $x \in \R^n$ and returns $y = Kx + e$ where $e$ is an entrywise non-negative error vector satisfying: $\norm{e}_2 \le \epsilon \norm{Kx}_2$.
\end{definition}

\begin{algorithm}[h]
	\caption{Kernel Noisy Power Method}
	\label{alg:power}
	{\bfseries input}: Error parameter $\epsilon \in (0,1)$. Iteration count $I$. $\epsilon^2/12$-non-negative approximate MVM algorithm (Def. \ref{def:mvm}) $\mathcal{K}(\cdot)$ for nice kernel matrix $K \in \R^{n \times n}$. \\
	{\bfseries output}: $z \in \R^n$ with $\norm{z}_2 = 1$ 
	\begin{algorithmic}[1]
	\STATE Initialize $z_0 \in \R^n$ with $z_0(i) = \frac{1}{\sqrt{n}}$ for all $i$.
	\STATE Initialize $\lambda := 0$.
	\FOR{$i = 0$ to $I$}
		\STATE $z_{i+1} := \mathcal{K}(z_{i})$.
		\IF{$z_{i}^T z_{i+1} > \lambda$}
		\STATE $z := z_{i}$.
		\STATE $\lambda := z_i^T z_{i+1}$.
		\ENDIF
		\STATE $z_{i+1} := z_{i+1}/\norm{z_{i+1}}_2.$
	\ENDFOR \\
	\STATE \textbf{return} $z$.
	\end{algorithmic}
\end{algorithm}


\begin{theorem}
\label{t:noisePower}
The kernel noisy power method (Algorithm \ref{alg:power}) run for 
     $I =  O \left (\frac{\log(n/\epsilon)}{\epsilon} \right)$ iterations outputs 
     a unit vector $z$ with $z^T K z \ge (1-\epsilon) \cdot \lambda_1(K)$.
\end{theorem}
\begin{proof}
Let $V \Lambda V^T = K$ be $K$'s eigendecomposition. $\Lambda$ is diagonal containing the eigenvalues in decreasing order $\lambda_1 \ge \ldots \ge \lambda_n \ge 0$. $V$ is orthonormal, with columns equal to the corresponding eigenvectors of $K$: $v_1,\ldots,v_n$. Let $m$ be the largest index such that $\lambda_m \ge (1-\epsilon/4) \cdot \lambda_1$.

Let $c_i =  V^T z_i$ be the $i^{th}$ iterate, written in the eigenvector basis. Let $c_{i,m}$ be its first $m$ components and $c_{i,n-m}$ be the last $n-m$ components. We will argue that for $I = O \left (\frac{\log(n/\epsilon)}{\epsilon} \right)$ iterations, there is at least one iteration $i \le I$ where $\norm{c_{i,m}}_2^2 \ge (1-\epsilon/4)$, and so $z_i$ aligns mostly with large eigenvectors. Formally this gives
\begin{align*}
z_i^T K z_i = c_i^T \Lambda c_i &\ge \norm{c_{i,m}}_2^2 \cdot (1-\epsilon/4) \cdot \lambda_1 \\
&\ge (1-\epsilon/4)^2 \cdot\lambda_1\\
&\ge (1-\epsilon/2) \cdot \lambda_1.
\end{align*}
Further, when we check to set $z := z_i$ at line (4) we have
\begin{align*}
z_{i}^T z_{i+1} = z_{i}^T \mathcal{K} z_{i}  = z_{i}^T K z_{i} + z_{i}^T e.
\end{align*}
Since $z_0$, $K$, and $e$ are all entrywise non-negative, $z_i$ is entrywise non-negative for all $i$. Thus, $ z_{i}^T e \ge 0$. By our bound on $\norm{e}_2$ we also have $z_{i}^T e \le \norm{e}_2 \le \epsilon^2/12 \cdot \norm{K z_{i-1}}_2 \le \epsilon^2/12 \cdot \lambda_1$.
Overall this gives
\begin{align*}
z_{i}^T K z_{i}  \le z_{i}^T z_{i+1} \le z_{i}^T K z_{i}  + \epsilon^2/12 \cdot \lambda_1.
\end{align*}
So, if there is an $i$ with $\norm{c_{i,m}}_2^2 \ge (1-\epsilon/4)$ and thus $z_i^T K z_i \ge (1-\epsilon/2) \cdot \lambda_1$, we will not output $z = z_j$ with  $z_j^T K z_j \le (1-\epsilon/2- \epsilon^2/12) \cdot \lambda_1 > (1-\epsilon) \cdot \lambda_1$, ensuring our final error bound.

%
%
To prove that there is some iterate with $\norm{c_{i,m}}_2^2 \ge (1-\epsilon/4)$, since $\norm{z_i}_2^2 = \norm{c_i}_2^2 = 1$, it suffices to argue that 
$\norm{c_{i,n-m}}_2^2 \le \epsilon/4$. 
Assume for the sake of contradiction that for all $i \le I$ we have $\norm{c_{i,n-m}}_2^2 > \epsilon/4$. Under this assumption we can argue that $c_i(1)^2$ grows significantly with respect to $\norm{c_{i,n-m}}_2^2$ in each step. Specifically, we can show by induction that $\frac{c_i(1)}{\norm{c_{i,n-m}}_2} \ge \frac{(1+\epsilon/6)^i}{\sqrt{n}}$. This gives a contradiction since for $I = O(\log(n/\epsilon)/\epsilon)$ it would imply that $1 \ge c_I(1)^2 \ge \frac{4}{\epsilon} \norm{c_{I,n-m}}_2^2$ and thus we must have $\norm{c_{I,n-m}}_2^2 < \epsilon/4$. This contradiction proves the theorem.



\smallskip 

\noindent\textbf{Base case.} Initially, $z_0$ has all entries equal to $1/\sqrt{n}$. Thus
\begin{align*}
\frac{c_0(1)}{\norm{c_{0,n-m}}_2} \ge c_0(1) = v_1^T z_0 &= \frac{1}{\sqrt{n}} \sum_{j=1}^n v_1(j)\\
&= \frac{1}{\sqrt{n}} \norm{v_1}_1 \ge \frac{1}{\sqrt{n}},
\end{align*}
where we use that $v_1$ is a non-negative unit vector by the Perron-Frobenius theorem so $\sum_{j=1}^n v_1(j) \ge \norm{v_1}_2 = 1$.

\smallskip 

\noindent\textbf{Inductive step.} 
Assume inductively that $\frac{c_i(1)}{\norm{c_{i,n-m}}_2} \ge \frac{(1+\epsilon/6)^i}{\sqrt{n}}$. Before normalization at step (7) we have $z_{i+1} = Kz_i + e$. 
Normalization doesn't affect the ratio between $c_{i+1}(1)$ and $ \norm{c_{i+1,n-m}}_2$ -- thus we can ignore this step. 

For all $j \in 1,\ldots,n$ we have $c_{i+1}(j) = \lambda_j \cdot c_{i}(j) + v_j^T e$. Since both $e$ and $v_1$ are all non-negative, this gives $c_{i+1}(1) \ge \lambda_1 \cdot c_{i}(1)$. Further, by triangle inequality and the fact that for $j > m$, $\lambda_j < (1-\epsilon/4)$,
\begin{align*}
\norm{c_{i+1,n-m}}_2 &\le (1-\epsilon/4) \lambda_1 \cdot \norm{c_{i,n-m}}_2 + \norm{e}_2\\
& \le (1-\epsilon/4) \lambda_1 \cdot \norm{c_{i,n-m}}_2 + \epsilon^2/12 \cdot \norm{Kz_i}_2\\
& \le (1-\epsilon/2) \lambda_1 \cdot \norm{c_{i,n-m}}_2 + \epsilon^2/12 \cdot \lambda_1.
\end{align*}
By our assumption (for contradiction) that $\norm{c_{i,n-m}}_2 \ge \epsilon/4$ we then have
\begin{align*}
\norm{c_{i+1,n-m}}_2 &\le \left (1-\epsilon/2 + \frac{\epsilon^2/12}{\epsilon/4} \right) \norm{c_{i,n-m}}_2 \cdot \lambda_1\\
&\le (1-\epsilon/6) \cdot \norm{c_{i,n-m}}_2 \cdot \lambda_1.
\end{align*}
Overall, this gives that
\begin{align*}
\frac{c_{i+1}(1) }{\norm{c_{i+1,n-m}}_2} &\ge \frac{\lambda_1 \cdot c_i(1)}{(1-\epsilon/6) \cdot  \norm{c_{i,n-m}}_2 \cdot \lambda_1}\\
&\ge (1+\epsilon/6) \cdot \frac{c_i(1)}{ \norm{c_{i,n-m}}_2} \ge \frac{(1+\epsilon/6)^{i+1}}{\sqrt{n}},
\end{align*}
by our inductive assumption. This gives our contradiction, completing the proof.
\end{proof}

\subsection{Approximate Kernel Matrix Vector Multiplication}\label{sec:mvm}

We next show how to use fast a KDE algorithm to instantiate the non-negative approximate MVM primitive (Definition \ref{def:mvm}) required by Algorithm \ref{alg:power}. A similar approach was taken by \citet{charikar2020kernel}. We provide our own analysis, which applies black box to any KDE algorithm.

\begin{theorem}\label{thm:mvm} Let $k$ be a nice kernel (Def. \ref{def:nice}) admitting $O(d m/\mu^p \epsilon^2)$ time approximate kernel density evaluation (Def. \ref{def:fast}). Let $K \in \R^{n \times n}$ be the associated kernel matrix for $n$ points in $d$ dimensions. There is an $\epsilon$-non-negative approximate MVM algorithm for $K$ running in time $O \left(\frac{dn^{1+p} \log(n/\epsilon)^p}{\epsilon^{3+2p}} \right )$.
\end{theorem}

Combined with Theorem \ref{t:noisePower}, Theorem \ref{thm:mvm} immediately gives our final subquadratic time eigenvalue approximation result:
\begin{corollary}\label{thm:powerMethodFinal} Let $k$ be a nice kernel (Def. \ref{def:nice}) admitting $O(d m/\mu^p \epsilon^2)$ time approximate kernel density evaluation (Def. \ref{def:fast}). Let $K \in \R^{n \times n}$ be the associated kernel matrix for $n$ points in $d$ dimensions. There is an algorithm running in time $O \left(\frac{dn^{1+p} \log(n/\epsilon)^{2+p}}{\epsilon^{7+4p}} \right )$, which outputs a unit vector $z$ with $z^T K z \ge (1-\epsilon) \cdot \lambda_1(K)$
\end{corollary}
\begin{proof}
Algorithm \ref{alg:power} requires $I = O \left (\frac{\log(n/\epsilon)}{\epsilon} \right)$ approximate MVMs, each with error parameter $\epsilon^2/12$. By Theorem \ref{thm:mvm}, each matrix vector multiply requires time $O \left(\frac{dn^{1+p} \log(n/\epsilon)^{1+p}}{\epsilon^{6+4p}} \right )$. Multiplying by $I$ gives the final bound.
\end{proof}

\begin{proof}[Proof of Theorem \ref{thm:mvm}]
We seek an algorithm that computes $Kx + e$ where $e$ is a non-negative error vector with $\norm{e}_2 \le c\epsilon \norm{Kx}_2$ for some fixed constant $c$. Note that we can always scale $x$ so that $\norm{x}_2 = 1$, and then scale back after multiplication, without effecting the error $\epsilon$. Thus we assume going forward that $\norm{x}_2 = 1$. We can also adjust $\epsilon$ by a constant factor to have error bounded by $\epsilon \norm{Kx}_2$ rather than $c\epsilon \norm{Kx}_2$. Since $K$ has all non-negative entries and ones on the diagonal,
\begin{align*}
\norm{Kx}_2^2 = \norm{x}_2^2 + \norm{(K-I)x}_2^2 + 2 x^T (K-I)x \ge \norm{x}_2^2,
\end{align*}
since $K-I$ also have all non-negative entries. Thus $\norm{Kx}_2 \ge 1$.

\smallskip

\noindent \textbf{Rounding $x$:} 
We split the entries of $x$ into $b = c_1 \left (\frac{\log(n/\epsilon)}{\epsilon} \right )$ buckets consisting of values lying in the range $[(1-\epsilon/2)^{i},(1-\epsilon/2)^{i-1}]$ for $i = 1,...,b$ for some fixed constant $c_1$. Let $\bar x$ have all values in bucket $i$ rounded to $(1-\epsilon/2)^{i-1}$. This rounding increases each entry in $x$ by at most a $\frac{1}{1-\epsilon/2} \le 1 + \epsilon$ multiplicative factor, and so increases all values in $Kx$ also by at most a $1+ \epsilon$ multiplicative factor. Thus, $\norm{K \bar x - K x}_2 \le \epsilon \norm{Kx}_2$. By triangle inequality, it suffices to compute $z = K \bar x + e$ where $e$ is non-negative and $\norm{e}_2 \le c \epsilon \norm{K \bar x}_2$ for some constant $c$.

Let $\bar x_i \in \R^n$ be $\bar x$ with only the entries in bucket $i$ kept and the rest set to zero. Let $\bar x_{b+1} \in \R^n$ be the set of entries not falling into any bucket -- note that these entries have not been rounded. We have $K \bar x = \sum_{i=1}^{b+1} K \bar x_i$. Thus, to achieve our final error bound, it suffices to return $z = \sum_{i=1}^{b+1} z_i$ where $z_i = K \bar x_i + e_i$, for non-negative error vector $e_i$ with $\norm{e_i}_2 \le \frac{\epsilon}{b+1} \norm{K \bar x}_2$. This gives that $z = K \bar x + e$ where $e = \sum_{i=1}^{b+1} e_i$ is non-negative  and by triangle inequality has $\norm{e}_2 \le \epsilon \norm{K \bar x}_2$, satisfying the required guarantees.

\medskip

\noindent \textbf{Remainder Entries:} We first focus on the entries not lying in any bucket, $\bar x_{b+1}$.
For large enough $c_1$ (recall that $b = c_1 \left (\frac{\log(n/\epsilon)}{\epsilon} \right )$), $\bar x_{b+1}$ only includes entries with value $ < \frac{\epsilon}{(b+1)n^{3/2}}$. Since all entries of $K$ are bounded by $1$, all entries of $K \bar x_{b+1}$ are bounded by $\frac{\epsilon}{(b+1) \sqrt{n}}$. Thus, if we let $z_{b+1}$ have value $\frac{\epsilon}{(b+1) \sqrt{n}}$ in each entry, we have $z_{b+1} = K \bar x_{b+1} + e$ where $e$ is non-negative and $\norm{e}_2 \le \frac{\epsilon}{b+1} \le \frac{\epsilon}{b+1} \norm{K \bar x}_2$, as required.

\medskip

\noindent\textbf{Estimation within a Bucket:} We next consider approximating $K \bar x_i$, which amounts to kernel density evaluation with query points corresponding to all $n$ points in our data set and target points $\mathcal{X}_i$ corresponding to the non-zero entries of $\bar x_i$. We have 
\begin{align*}[K \bar x_i](j) = (1-\epsilon/2)^{i-1} \sum_{x \in \mathcal{X}_i} k(x,x_j) = |\mathcal{X}_i| (1-\epsilon/2)^{i-1} \cdot \frac{1}{|\mathcal{X}_i|}   \sum_{x \in \mathcal{X}_i} k(x,x_j).
\end{align*}
If we round any $[K \bar x_i](j)$ with value $\le \frac{\epsilon}{(b+1) \sqrt{n}}$ up to $\frac{\epsilon}{(b+1) \sqrt{n}}$, this will affect our final error by at most $\frac{\epsilon}{(b+1)} \le \frac{\epsilon}{b+1} \norm{K \bar x}_2$. Thus, to achieve our designed error, we set our 
KDE relative error to $\epsilon$ and minimum density value to $\mu = \frac{\epsilon}{(b+1) \sqrt{n}} \cdot \frac{1}{|\mathcal{X}_i| (1-\epsilon/2)^{i-1}}$. We multiply each estimate by a $1/(1-\epsilon)$ factor to ensure that it is an overestimate of the true entry in $K \bar x_i$. Since $\norm{\bar x}_2 \le (1+\epsilon) \norm{x}_2 = (1+\epsilon)$, we have $|\mathcal{X}_i| \le \min \left (n,\frac{(1+\epsilon)^2}{(1-\epsilon/2)^{2(i-1)}}\right)$. Thus, 
\begin{align*}\frac{1}{|\mathcal{X}_i| (1-\epsilon/2)^{i-1}}\ge \max \left (\frac{1}{n (1-\epsilon/2)^{i-1}} , \frac{(1-\epsilon/2)^{i-1}}{(1+\epsilon)^2} \right ) \ge \frac{1}{4\sqrt{n}}.
\end{align*}
In turn, this gives $\mu = \frac{\epsilon}{(b+1) \sqrt{n}} \cdot \frac{1}{|\mathcal{X}_i| (1-\epsilon/2)^{i-1}} \ge \frac{\epsilon}{(b+1) 4n}$.

Finally, we plug in our KDE runtime of $O(dm/\mu^p \epsilon^2)$. For each bucket, we must process $|\mathcal{X}_i|$ points to evaluate the density against. Since $\sum_{i=1}^b |\mathcal{X}_i| \le n$ the total runtime here is $O(dn/\mu^p \epsilon^2)$. We then must evaluate the density of all $n$ points against the points in each bucket, requiring total time $O(b \cdot dn/\mu^p \epsilon^2)$. Recalling that $b = O \left (\frac{\log(n/\epsilon)}{\epsilon}\right )$ and $\mu = \frac{\epsilon}{(b+1) 4n}$ gives total runtime:
\begin{align*}O \left(\frac{dn^{1+p} b^{1+p}}{\epsilon^{2+p}} \right ) = O \left(\frac{dn^{1+p} \log(n/\epsilon)^{1+p}}{\epsilon^{3+2p}} \right ).
\end{align*}
\end{proof}

\subsection{Lower Bound}

It is easy to see that $\Omega(dn)$ time is necessary to estimate $\lambda_1(K)$ to even a constant  factor. Thus, for $\epsilon = \Theta(1)$, the runtime of Corollary \ref{thm:powerMethodFinal} is tight, up to an $n^p \log(n)^{2+p}$ factor.

\begin{theorem}\label{thm:eigLower}
    Consider any nice kernel $k$ such that $k(x,y) \to 0$ as $\|x-y\| \to \infty$. Let $K \in \R^{n \times n}$ be the associated kernel matrix of $n$ points $x_1,\ldots,x_n \in \R^d$. Estimating $\lambda_1(K)$ to any constant factor requires $\Omega(nd)$ time.
 \end{theorem}
 \begin{proof}
 Let $c$ be any constant. 
 Consider two input cases. In the first, no two points in $x_1,\ldots,x_n$  are identical. In the second, a random set of $c$ points are exact duplicates of each other. Scaling up these point sets by an arbitrarily large constant value, we can see that their kernel matrices are arbitrarily close to $K = I$ in the first case and $K = I + E$ in the second, where $E$ has a 1 at positions $(i,j)$ and $(j,i)$ if $i,j$ are in the duplicate set, and zeros everywhere else. We can check that $\lambda_1(I) = 1$ while $\lambda_1(I+E) = c$. Thus, to approximate the top eigenvalue up to a $c$ factor we must distinguish the two cases. If we read $o(n/c)$ points, then with good probability we will see no duplicates. Thus, we must read $\Omega(n/c)$ points, requiring $\Omega(nd/c)$ time.
 \end{proof}

\section{Empirical evaluation}\label{sec:exp}
In this section we empirically evaluate the kernel noisy power method (Algo. \ref{alg:power}) for approximating the top eigenvalue of the kernel matrix $K\in\R^{n\times n}$ associated with an input dataset of $n$ points. 

The reason we focus our evaluation on approximating the top eigenvector (Section \ref{sec:topeigvec}), and not on approximating the kernel matrix sum (Section \ref{sec:kernelsum}) is that the latter problem admits a fast algorithm by vanilla random sampling (in time $O(n)$, see Claim \ref{clm:unisampling}), which is very fast in practice. We observed this baseline method obtains comparable empirical running times to our method. We thus focus our experiments on the kernel power method. All previous algorithms for approximating the top eigenvector have running time $\Omega(n^2)$, which is very slow in practice even for moderate $n$, raising a stronger need for empirical improvement.

For evaluating the kernel noisy power method, we use the Laplacian kernel, $k(x,y)=\exp(-\norm{x-y}_1/\sigma)$. It is a nice kernel (by Def.~\ref{def:nice}), and furthermore, a Fast KDE implementation for it (as per Def.~\ref{def:fast}) with $p=0.5$ was recently given in~\cite{backurs2019space}, based on the Hashing-Based Estimators (HBE) technique of~\cite{charikar2017hashing}. This can be plugged into Corollary~\ref{thm:powerMethodFinal} to obtain a provably fast and accurate instantiation of the kernel noisy power method. The resulting algorithm is referred to in this section as~\textbf{KNPM}.

\paragraph{Evaluated methods.}
We compare KNPM to two baselines: The usual (\textbf{Full}) power method, and a~\textbf{Uniform} noisy power method. The latter is similar to KNPM, except that the KDE subroutine in Corollary~\ref{thm:powerMethodFinal} is evaluated by vanilla uniform sampling instead of a Fast KDE.

In more detail, denote by $UniKDE(y,X)$ a randomized algorithm for $KDE(y,X)$ on a point $y$ and a pointset $X$, that draws a uniformly random sample $X'$ of $X$ and returns the mean of $k(y,x)$ over all $x\in X'$. By Bernstein's inequality, if the sample size is $\Omega(1/(\mu\epsilon^2))$ then this returns a $(1\pm\epsilon)$-approximation of $KDE(y,X)$
in time $O(d/(\mu\epsilon^2))$. Therefore, it is a Fast KDE algorithm as defined in Definition \ref{def:fast}, with $p=1$. This algorithm is indeed used and analyzed in many prior works on KDE approximation. Since Algorithm \ref{alg:power} reduces the kernel power method to a sequence of approximate KDE computations (cf.~Corollary \ref{thm:powerMethodFinal}), we may use $UniKDE$ for each of them, thus obtaining the natural baseline we call Uniform. While it asypmtotically does not lead to sub-quadratic running time, it empirically performs significantly better than the full power method, as our experiments will show.

\paragraph{Evaluation metrics.}
The computational cost of each algorithm is measured by the number of kernel evaluations performed (i.e., how many entries of $K$ are computed). 
Computed kernel values are not carried over across iterations. 
The full power methods computes the entire matrix  in each iteration, since it is too large to be stored in memory. The other two methods use a small sample of entries.
We use the number of kernel evaluations as a proxy for the running time since:
\begin{itemize}
\item The three algorithms we evaluate compute power method by a sequence of kernel computations, differing in the choice of points for evaluation (Full computes $k(x,y)$ for all $x,y$, while Uniform and KNPM choose pairs at random according to their different sampling schemes). Thus this measure of efficiency allows for a direct comparison between them.
\item The measure is software and architecture-free, unaffected by access to specialized libraries (e.g., BLAS, MATLAB) or hardware (e.g., SIMD, GPU). This is important with linear algebraic operations, which behave very differently in different environments, resulting in  artifacts when measuring runtimes.
\item Compatibility with prior literature (e.g., \cite{backurs2019space}).
\end{itemize}
Nonetheless, we believe that methodologically sound wall-clock time experiments would be valuable, and we leave this for future work.

The accuracy of each method is evaluated in each iteration by the relative error, $1-z^TKz/\lambda_1(K)$, where $z$ is the unit vector computed by the algorithm in that iteration, and $\lambda_1(K)$ is the true top eigenvalue. $\lambda_1(K)$ is computed by  letting the full power method run until convergence. This error measure corresponds directly to $\epsilon$ from Corollary~\ref{thm:powerMethodFinal}.

\vspace{-.2em}
\paragraph{Datasets.}
We use classes of the Forest Covertype dataset~\citep{blackard1999comparative}, which is a 54-dimensional dataset often used to evaluate kernel methods in high dimensions \citep{siminelakis2019rehashing,backurs2019space}. We use 5 of the 7 classes (namely classes 3--7), whose sizes range from 2.7K to 35.7K points. We have omitted the two larger classes since we could not compute an accurate groundtruth $\lambda_1(K)$ for them, and hence could not measure accuracy. We also use the full training set of the MNIST dataset (60K points in 784 dimensions).


\vspace{-.2em}
\paragraph{Parameter setting.}
We use bandwidth $\sigma=0.05$ (other choices produce  similar results). The full power method has no parameters. The Uniform and KNPM  methods each have a single parameter that governs the sampling rate. For both, we start with a small sampling rate, and gradually increase it by multiplying by $1.1$ in each iteration. In this way, the approximate matrix multiplication becomes more accurate as the method converges closer to the true top eigenvalue.

\vspace{-.2em}
\paragraph{Results.} All results are reported in Figure~\ref{fig:knpm} on the following page. Both the Uniform and KNPM variants of the noisy power method give a much better tradeoff in terms of accuracy vs. computation than the full power method. Additionally,  KNPM consistently outperforms Uniform.

\vspace{-.5em}
\section{Conclusion}
We have shown that fast kernel density evaluation methods can be used to give much faster algorithms for approximating the kernel matrix sum and its top eigenvector. Our work leaves open a number of directions. For  top eigenvector computation -- it is open if the gaps between the linear in $n$ lower bound of Theorem \ref{thm:eigLower} and our slightly superlinear runtimes for the Gaussian and exponential kernels can be closed. Extending our techniques to approximate the top $k$ eigenvectors/values in subquadratic time would also be very interesting, as this is a key primitive in kernel PCA and related methods. 
Finally, it would be interesting to identify other natural kernel matrix problems that can be solved in sublinear or subquadratic time using fast KDE methods. Conversely, one might hope to prove lower bounds ruling this out. Lower bounds that hold  for  error $\epsilon = \Theta(1)$ would be especially interesting -- known lower bounds against e.g., subquadratic time algorithms for the kernel sum, only hold when  high accuracy  $\epsilon=\exp(-\omega(\log^2 n))$ is demanded. 

\begin{figure*}[ht] 
  \begin{subfigure}
    \centering
    \includegraphics[width=0.5\linewidth]{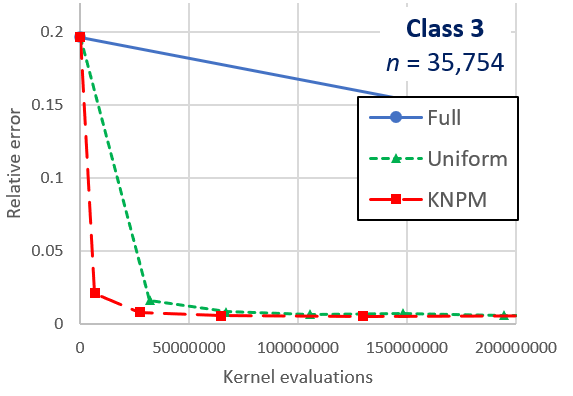} 
    \vspace{4ex}
  \end{subfigure}
  \begin{subfigure}
    \centering
    \includegraphics[width=0.5\linewidth]{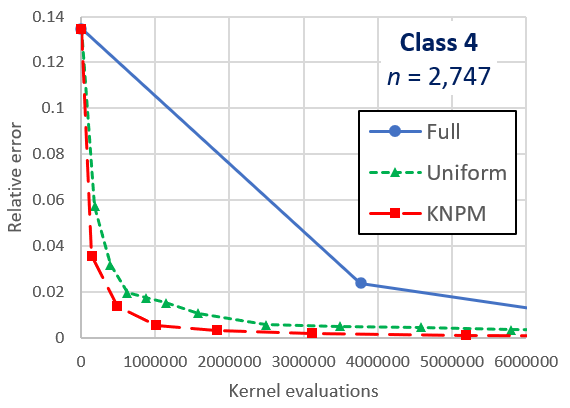} 
    \vspace{4ex}
  \end{subfigure} 
  \begin{subfigure}
    \centering
    \includegraphics[width=0.5\linewidth]{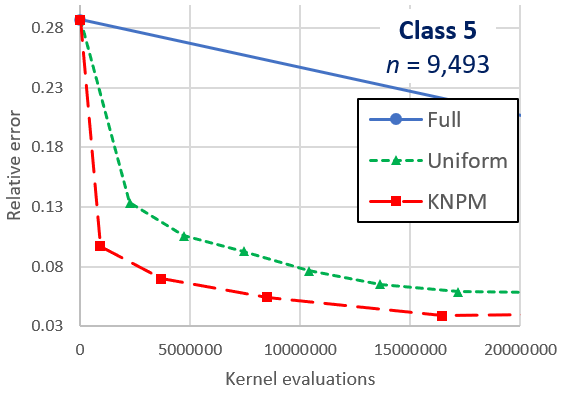} 
  \end{subfigure}
  \begin{subfigure}
    \centering
    \includegraphics[width=0.5\linewidth]{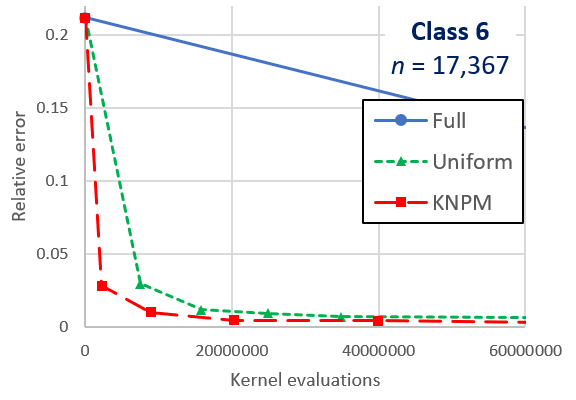} 
  \end{subfigure} 
    \begin{subfigure}
    \centering
    \includegraphics[width=0.5\linewidth]{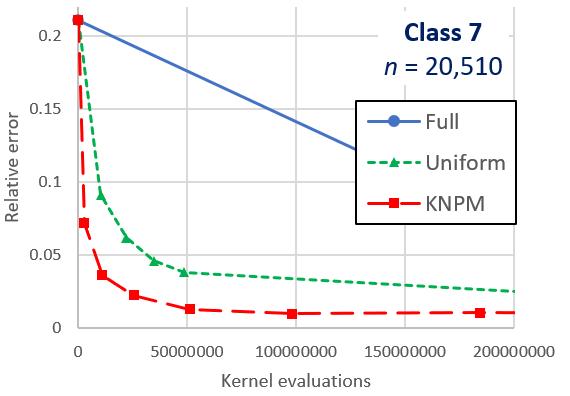} 
  \end{subfigure}
  \begin{subfigure}
    \centering
    \includegraphics[width=0.49\linewidth]{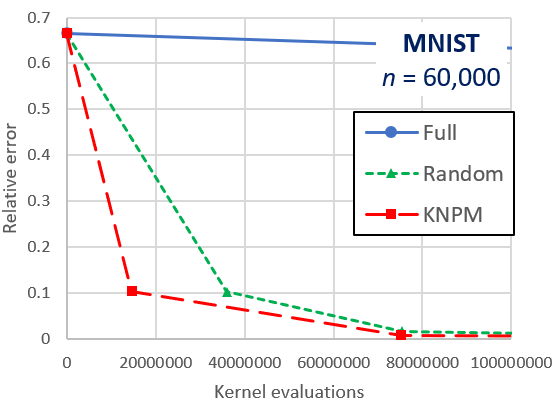} 
  \end{subfigure} 
  \caption{Results for the Full, Uniform and Kernel Noisy variants of the power method, on classes 3--7 of the Forest Covertype dataset, and on the MNIST dataset. We can see that the noisy power method implemented with KDE (KNPM) achieves a significantly better tradeoff between accuracy  and number of kernel evaluations, as compared to the baselines.}
  \label{fig:knpm} 
\end{figure*}


\section*{Acknowledgments} This  research  was  supported  in  part  by  the  NSF  TRIPODS  program  (awards  CCF-1740751  and DMS-2022448);  NSF  award  CCF-2006798; NSF award CCF-2006806; MIT-IBM Watson collaboration; a  Simons Investigator Award; NSF awards IIS-1763618 and CCF-2046235; and an Adobe Research grant.
\clearpage

\bibliographystyle{plainnat}
\bibliography{kernelApprox}

\clearpage
\appendix

\section{Simpler Algorithm for the Fast Gauss Transform}

In this section we provide a simpler version of the algorithm of~\cite{greengard1991fast} for approximate Gaussian kernel density estimation in low dimensions. This version uses ideas that are similar to those in the original algorithm, but reduces it to the necessary essentials, removing e.g., the use of Hermite polynomials. 

\begin{theorem}
    Given $n$ points $p_1, \ldots, p_n$, we can preprocess them in $n \log(1/\eps)^{O(d)}$ time so that we can answer queries of the form ${1 \over n}\sum_i k(q,p_i)$ for $k(q,p)=\exp(-\|q-p\|^2)$ within an additive factor of $\eps$ in $\log(1/\eps)^{O(d)}$ query time.
\end{theorem}
\begin{proof}
    We first note that, if $\|q-p_i\|^2 \ge \log(1/\eps)$, then we can discard such points and this changes the average of the kernel values by at most an additive factor of $\eps$.
    
    In the preprocessing step we partition the space into $d$-dimensional hypercubes of \emph{diameter} $\log(1/\eps)$. When a query point $q$ comes, we only examine those hypercubes (and the points of the dataset within) that are at distance at most $5\log(1/\eps)$  from the query point $q$. This ensures that don't ignore any point $p_i$ that is at distance less than $\log(1/\eps)$ from $q$. Furthermore, the total volume of all inspected hypercubes is less than the volume of a $d$-dimensional Euclidean ball of radius $10\log(1/\eps)$. A simple volume argument shows that such a Euclidean ball can contain at most $2^{O(d)}$ hypercubes of diameter $\log(1/\eps)$. Thus, our algorithm will only examine at most $2^{O(d)}$ hypercubes. Since this is asymptotically negligible compared to the promised query time, from now on we will assume that all points $p_i$ are satisfy $\|q-p_i\|^2 \leq 10\log(1/\eps)$.
    
    We calculate the Taylor expansion of $\exp(-\|q-p_i\|^2)$. We denote $x=-\|q-p_i\|^2$ obtaining:
    $$
        \exp(-\|q-p_i\|^2)=\exp(x)=1 + x + x^2/2! + x^3 / 3! + \ldots
    $$
    Since $|x| \leq 10 \log(1/\eps)$, we can truncate the expansion after $O(\log(1/\eps))$ terms such that the truncation error is at most $\eps$. Consider a term $x^r/r!$ with $r \leq O(\log(1/\eps))$. Since $r \leq O(\log(1/\eps))$, it suffices to have an $r^{O(d)}$ time algorithm (to get the promised final runtime) to answer the queries of the form $\sum_i \|q-p_i\|^{2r}$. This is sufficient because then we can do separate queries for all $r\leq O(\log(1/\eps))$ and combine them according to the Taylor expansion to approximate $\exp(-\|q-p_i\|^2)$.
    
    To answer the queries of the form $\sum_i \|q-p_i\|^{2r}$ efficiently, it is sufficient to answer the queries of the form $\sum_i(q \cdot p_i)^r$ efficiently, where $q \cdot p_i = q(1) p_i(1) + \ldots + q(d) p_i(d)$ denotes the inner product between the two vectors.
    
    To achieve a fast query algorithm we observe that
    $$
        (q \cdot p_i)^r = ( q(1) p_i(1) + ... + q(d) p_i(d) )^r
    $$
    consists of $r^{O(d)}$ different monomials after opening the parentheses.
    Thus, we can write
    $$
        (q \cdot p_i)^r = X(q) \cdot Y(p_i)
    $$
    for some efficiently computable maps $X, Y : \R^d \to \R^{r^{O(d)}}$.
    We get that
    $$
        \sum_i (q \cdot p_i)^r = X(q) \cdot \sum_i Y(p_i).
    $$
    It remains to note that we can precompute quantities $\sum_i Y(p_i)$ for every hypercube in the preprocessing step and then examine the relevant hypercubes during the query step.
\end{proof}

    \paragraph{Note.} Similar algorithms work for kernels $\log\|q-p_i\|$, $1/(1+\|q-p_i\|^t)$ etc.:
    do Taylor expansion and precompute the quantities $\sum_i Y(p_i)$.


\end{document}

%% file: bib_macros.tex
  \usepackage{nth}
  \usepackage{intcalc}